\documentclass[a4paper,UKenglish]{lipics-v2018}
\usepackage{algorithm,algorithmicx,algpseudocode}
\usepackage{microtype}
\usepackage{tikz}
\usetikzlibrary{patterns}
\usepackage{thmtools,thm-restate}
\usepackage{xspace}

\theoremstyle{plain}
\newtheorem{fact}[theorem]{Fact}

\newcommand{\D}{\mathcal{D}}
\newcommand{\pred}{\mathsf{pred}}
\newcommand{\successor}{\mathsf{succ}}
\newcommand{\Oh}{\mathcal{O}}
\newcommand{\PrefixSearch}{\textsf{PrefixSearch}\xspace}
\newcommand{\Suff}{\mathrm{Suf}}
\newcommand{\eps}{\varepsilon}

\nolinenumbers

\title{Fast entropy-bounded string dictionary look-up with mismatches}
\author{Pawe\l{} Gawrychowski}{University of Wrocław, Wrocław, 50-137, Poland}{gawry@cs.uni.wroc.pl}{}{}
\author{Gad M. Landau}{University of Haifa, Haifa, 3498838, Israel}{landau@cs.haifa.ac.il}{}{}
\author{Tatiana Starikovskaya}{DIENS, \'{E}cole normale sup\'{e}rieure, PSL Research University, Paris, 75005, France}{tat.starikovskaya@gmail.com}{}{}

\titlerunning{Fast string dictionary look-up with mismatches}
\authorrunning{P. Gawrychowski, G.M. Landau, T. Starikovskaya} 
\Copyright{P. Gawrychowski, G.M. Landau, T. Starikovskaya}
\subjclass{}

\EventEditors{Igor Potapov, Paul Spirakis, and James Worrell}
\EventNoEds{3}
\EventLongTitle{43rd International Symposium on Mathematical Foundations of Computer Science (MFCS 2018)}
\EventShortTitle{MFCS 2018}
\EventAcronym{MFCS}
\EventYear{2018}
\EventDate{August 27--31, 2018}
\EventLocation{Liverpool, GB}
\EventLogo{}
\SeriesVolume{117}
\ArticleNo{66}

\begin{document}

\maketitle

\begin{abstract}
We revisit the fundamental problem of dictionary look-up with mismatches. Given a set (dictionary) of $d$ strings of length $m$ and an integer $k$, we must preprocess it into a data structure to answer the following queries: Given a query string $Q$ of length $m$, find all strings in the dictionary that are at Hamming distance at most $k$ from $Q$. Chan and Lewenstein (CPM 2015) showed a data structure for $k = 1$ with optimal query time $\Oh(m/w + occ)$, where $w$ is the size of a machine word and $occ$ is the size of the output. The data structure occupies $\Oh(w d \log^{1+\eps} d)$ extra bits of space (beyond the entropy-bounded space required to store the dictionary strings). In this work we give a solution with similar bounds for a much wider range of values $k$. Namely, we give a data structure  that has $\Oh(m/w + \log^k d + occ)$ query time and uses $\Oh(w d \log^k d)$ extra bits of space.   
\end{abstract}

\section{Introduction}
	\label{sec:introduction}
	The problem of dictionary look-up was introduced by Minsky and Papert in 1968 and is a fundamental task in many areas such as bioinformatics, information retrieval,  and web search. Informally, the task is to store a set of strings referred to as dictionary in small space to maintain the following queries efficiently: Given a query string, return all dictionary strings that are close to it under some measure of distance. In this work we focus on Hamming distance and exact solutions to the problem. Formally, the problem is stated as follows.

\textbf{Dictionary look-up with $k$ mismatches.} We are given a dictionary that is a set of $d$ strings of length $m$ and an integer $k > 0$. The task is to preprocess the dictionary into a data structure that maintains the following queries: Given a string $P$ of length $m$, return all the strings in the dictionary such that the distance between each of them and $P$ is at most $k$.

As a natural first step, much effort has been concentrated on the case $k =1$~\cite{Belazzougui2009FasterAS,Hon2011CompressedDM,Belazzougui2012CompressedSD,Brodal1996ApproximateDQ,Brodal2000ImprovedBF,Yao1997DictionaryLW,CL2015}.
We note that the structure of the problem in this case is very special. Namely, if two strings have Hamming distance at most one, then there is an integer $i$ such that their prefixes of length $i$ are equal and their suffixes of length $m-i-1$ are equal. Many existing solutions rely heavily on this property and cannot be extended to the case of arbitrary $k$. The first non-trivial solution for $k > 1$ was given in the seminal paper of Cole, Gottlieb, and Lewenstein~\cite{kerratatree}, who introduced a data structure called \emph{$k$-errata tree}. The $k$-errata tree requires $w \cdot \Oh(md + d\log^k d)$ bits of space and has query time $\Oh(m + \log^k d + occ)$, where $w$ is the size of a machine word and $occ$ is the size of the output. The subsequent work~\cite{Chan2006ALS,Chan2006CompressedIF,Lam2005ImprovedAS} mainly focused on improving the space complexity of the data structure.

Two works are of particular interest to us. For any $q = o(\log md)$ Belazzougui and Venturini~\cite{Belazzougui2012CompressedSD} showed a data structure for $k=1$ with query time $\Oh(m+occ)$ that uses $2md H_q + o(md) + 2d \log d$ bits of space, where $H_q$ is the $q$-th empirical entropy of the concatenation of all strings in the dictionary. It was followed by the work of Chan and Lewenstein~\cite{CL2015}, who improved the query time to $\Oh(m/w + occ)$, while using approximately the same amount of bits, $2md H_q + o(md) + \Oh( w d \log^{1+\eps} d)$. In the model of Chan and Lewenstein the size $\sigma$ of the alphabet is constant, the query string arrives in a packed form, meaning that each $w / \log \sigma$ letters are stored in one machine word, under the standard assumption $w = \Theta(\log md)$. The interest in this kind of bounds is explained by the fact that the value $md H_q$ is a lower bound to the output size of any compressor that encodes each letter of the dictionary strings with a code that only depends on the letter itself and on the $q$ immediately preceding letters~\cite{Manzini:2001:ABT:382780.382782}.

\subsection{Our contribution and techniques}
We investigate further this line of research and give a new data structure with similar bounds for a much wider range of values $k$. We adopt the model of Chan and Lewenstein and show a data structure for dictionary look-up with $k$ mismatches that has query time $\Oh(m/w + \log^k d + occ)$ and uses $2mdH_q + o(md) + \Oh(w d \log^{k} d)$ bits of space for all $d > 2$ (Theorem~\ref{th:main}). If in addition $k \le \log (m/w) / \log \log d$, the query time becomes $\Oh(m/w + occ)$, matching the query time of Chan and Lewenstein.

The basis of our data structure is the $k$-errata tree of Cole, Gottlieb, and Lewenstein~\cite{kerratatree}. We first introduce a small but important modification to this data structure that will allow us to reduce the time requirements for non-constant $k$. At a high level, the $k$-errata tree is a collection of compact tries, where each trie contains suffixes of a subset of strings in the dictionary. The query algorithm runs $\Oh(\log^k d)$ prefix search queries in the tries.  In Section~\ref{sec:LCP-linear} we show that the prefix search queries can be implemented in $\Oh(m/w)$ shared plus $\Oh(\log d)$ time per query using $\Oh(md)$ space beyond the space required by the $k$-errata tree. Next, in Section~\ref{sec:LCP-entropy} we show how to improve the space complexity to entropy-bounded. Our main contribution at this step is a new reduction from prefix search queries in the tries of the $k$-errata trees to prefix search queries on a compact trie containing only a subset of all suffixes of the dictionary strings. Finally, in Section~\ref{sec:removeloglog} we improve the $\Oh(\log d)$ time that we spend per each query to $\Oh(1)$ (amortised) time by a clever use of Karp-Rabin fingerprints, which gives the final result, Theorem~\ref{th:main}. We emphasize that we derandomize the query algorithm and that our data structure is deterministic, regardless the fact that we use Karp-Rabin fingerprints.

\subsection{Related work} 
Many of the works we cited above consider not only the Hamming distance, but also the edit distance. This is in particular true for $k = 1$, when the edit distance and Hamming distance are equivalent. Another interesting direction is heuristic methods for the Hamming and the edit distances which have worse theoretical guarantees but perform well in practice~\cite{CG2017,BHSH2007,Karch2010ImprovedFS,Mor:1982:HCM:358728.358752}. 
Finally, we note that the solutions discussed in this work are beneficial for low-distance regime, i.e. when $k = o(\log d)$. If $k = \omega (\log d)$, one should turn to approximate  approaches, such as locality-sensitive hashing (see~\cite{Andoni:2015:ODH:2746539.2746553} and references therein).

Several works have studied the question of developing efficient data structures for string processing when the query arrives in a packed form. In particular, Takuya et al. suggested a data structure called packed compact tries~\cite{packedtrie} to maintain efficient exact dictionary look-ups, and Bille, G{\o}rtz, and Skjoldjensen used similar technique to develop an efficient text index~\cite{packedtextindex}. 
	
\section{Preliminaries}
	\label{sec:preliminaries}
	We assume a constant-size integer alphabet $\{1, 2, \dots, \sigma\}$. A \emph{string} is a sequence   of letters of the alphabet. For a string $S = s_1 s_2 \dots s_m$ we denote its length $m$ by $|S|$ and its substring $s_i s_{i+1} \dots s_j$, where $1\le i < j \le m$, by $S[i,j]$. If $i = 1$, the substring $S[1,j]$ is referred to as \emph{prefix} of~$S$. If $j = m$, $S[i,m]$ is called a \emph{suffix} of $S$. We say that $S$ is given in a \emph{packed form} if each $w / \log \sigma$ letters of $S$ are stored in one machine word, i.e. $S$ occupies $\Oh(m/w)$ machine words in total. Given a string $S$ in packed form, we can access (a packed representation) of any $\Oh(w)$-length substring of $S$  in constant time using the shift operation.

A \emph{trie} is a basic data structure for storing a set of strings. A trie is a tree which has the following three properties:

\begin{enumerate}
\item Each edge is labelled by a letter of the alphabet;
\item Each two edges outgoing from the same node are labelled with different letters, and the edges are ordered by the letters;
\item Let the label of a node $u$ be equal to the concatenation of the labels of the edges in the root-to-$u$ path. For each string $S$ in the set there is a node of the trie such that its label is equal to $S$, and the label of each node is equal to a prefix of some string in the set. 
\end{enumerate}

At each node we store the set of ids of the strings that are equal to the node's label. The number of nodes in a trie can be proportional to the total length of the strings. To improve the space requirements, we replace each path of nodes with degree one and with no string ids assigned to them with an edge labelled by the concatenation of the letters on the edges in the path. The result is called a \emph{compact trie}. 
Each node of the trie is represented in the compact trie as well, some as nodes, and some as positions in the edges. We refer to the set of all nodes and the positions in the edges of the compact trie as positions. 

\begin{fact}
A compact trie containing $x$ strings has $\Oh(x)$ nodes.
\end{fact}

\section{The k-errata tree: Reminder and fix}
	\label{sec:k-errata}
	Our definition of the $k$-errata tree follows closely that of Cole, Gottlieb, and Lewenstein~\cite{kerratatree}, but as explained below we introduce an important fix to the original definition. We try to be as concise as possible, but we feel obliged to provide all the details both because we modify the original definition and because the details are important for our final result. 

\subparagraph{Intuition.} Let us explain the main idea first. Denote the given  dictionary of strings by $\D$. The $k$-errata tree for $\D$ is built recursively. We start with the compact trie $T$ containing all the strings in $\D$ and decompose it into heavy paths. 

\begin{definition}[\cite{heavypath}]
The \emph{heavy path} of $T$ is the path that starts at the root of $T$ and at each node $v$ on the path branches to the child with the largest number of leaves in its subtree (\emph{heavy} child), with ties broken arbitrarily. The heavy path decomposition is defined recursively, namely it is defined to be a union of the heavy path of $T$ and the heavy path decompositions of the subtrees of $T$ that hang off the heavy path. The first node in a heavy path is referred to as its \emph{head}.
\end{definition}

Recall that our task is to find all strings in $\D$ such that the Hamming distance between them and the query string $P$ is at most $k$. As a first step, we find the longest path that starts at the root of $T$ and is labelled with a prefix of $P$. Let this path trace heavy paths $H_1, H_2, \dots, H_j$, leaving the heavy path $H_i$ at a position $u_i$ of $T$, $1 \le i \le j$. We can partition all the strings in $\D$ into three categories: 

\begin{enumerate}
\item \label{tp:i} Strings diverging off a heavy path $H_i$ at some node $u$, where $u$ is located above $u_i$;
\item \label{tp:ii} Strings in the subtrees of $u_i$'s children that diverge from the heavy path $H_{i+1}$, for $1 \le i < j$;
\item  \label{tp:iii} Strings in the subtree rooted at $u_j$.
\end{enumerate}

Consider the set of strings in $\D$ that diverge from a heavy path $H_i$ at a node $u$. They necessarily have their first mismatch with $P$ there. The first idea is that we can fix that mismatch in each of the strings (decreasing the Hamming distance between them and $P$ by one), and then run a dictionary look-up with $(k-1)$ mismatches on the resulting set of strings. The second idea is that running an independent dictionary look-up query for each node in each heavy path is expensive, so we introduce a grouping on the nodes that reduces the number of queries to logarithmic.

\subparagraph{Data structure.}
We assign each string in $\D$ a credit of $k$ mismatches and start building the $k$-errata tree in a recursive manner. First, we build the compact trie $T$ for the dictionary $\D$. For each leaf of $T$ we store the ids of the dictionary strings equal to the leaf's label ordered by the mismatch credits. Second, we decompose $T$ into heavy paths. For each node of $T$ we store a pointer to the heavy path it belongs to, and for each heavy path we store a pointer to its head. We will now make use of weight-balanced trees, defined in analogy with weight-balanced search trees. 

\begin{definition}
A weight-balanced tree with leaves of weights $w_1, w_2, \dots, w_h$ (in left-to-right order) is a ternary tree. We build it recursively top-to-down. Let $\mu$ be the smallest index such that $w_1  + \dots + w_\mu > (w_1  + \dots + w_h)/2$. Then the left subtree hanging from the root is a weight-balanced tree with leaves of weight $w_1, w_2, \dots, w_{\mu-1}$, the middle contains one leaf of weight $w_\mu$, and the right subtree is a weight-balanced tree with leaves of weight $w_{\mu+1}, \dots, w_h$.
\end{definition} 
 
We build two sets of $(k-1)$-errata trees for each heavy path $H$ of $T$. We call the trees in the first set vertical, and in the second set horizontal, according to the way we construct them.

We first explain how we build the vertical $(k-1)$-errata trees. Suppose that $H$ contains nodes $v_1, v_2, \dots, v_h$, and the weight $w_i$ of a node $v_i$ is the number of strings that diverge from $H$ at $v_i$. As a preliminary step, we build a weight-balanced tree $WBT(H)$ on the nodes in $H$. Consider a node of $WBT(H)$ containing $v_i, v_{i+1}, \dots, v_j$ in its subtree. Let $\delta$ be the length of the string $S$ written on the path from the head of $H$ to $v_{j}$, and $a$ be the first letter on the edge from $v_{j}$ to $v_{j+1}$. We build a new set of strings as follows: For each node $v_\ell$, $i \le \ell \le j$, we take each string that diverges from the path $H$ at $v_\ell$, cut off its prefix of length $\delta+1$, and decrease the credit of the string by the number of mismatches between the cut-off prefix and $S \circ a$ (the string $S$ appended with the letter $a$). If the credit of a string becomes negative, we delete it. Finally, we build the $(k-1)$-errata tree for each of the newly created sets of strings.

We now explain how we build the horizontal $(k-1)$-errata trees. We repeat the following for each node $v_j \in H$. Let $\delta$ be the length of the label of $v_j$. Consider the set of all children of $v_j$ except for the node $v_{j+1}$ (the child of $v_j$ that belongs to $H$). For each child $v$ in this set, we build a new set of strings as follows: We take each string that ends below $v$, cut off its prefix of length $\delta+1$, and decrease the credit of the string by $1$. Similar to above, if the credit of a string becomes negative, we delete it. We define the weight of each child as the number of strings in the corresponding set. Next, we build the weight-balanced tree on the set of the children, and for each node of the tree consider a set of strings that is a union of the sets of strings below it. Finally, we build the $(k-1)$-errata tree for each of these sets of strings. 

\begin{remark}
Our modification to the original definition is that we truncate the strings and store the mismatch credits. Because of that, all the strings we work with are suffixes of the dictionary strings, which allows us to process them efficiently.
\end{remark}

\subparagraph{Queries.} A dictionary look-up with $k$ mismatches for a string $P$ is performed in a recursive way as well. For the purposes of recursion, we introduce an extra parameter, $\mu$, and allow to run dictionary look-ups with mismatches from any position $u$ of a trie of the $k$-errata tree. We will make use of a procedure called \PrefixSearch: Given a string and a position $u$ of a trie, \PrefixSearch returns the longest path starting at $u$ that is labelled by a prefix of the query string. 

Suppose we must answer a dictionary look-up with $k$ mismatches for a string $P$ that starts at a position $u$.
We initialize $\mu = 0$. If $k = 0$, we run a \PrefixSearch to find a path in $T$ labelled by $P$. If such a path exists, we output all the dictionary strings assigned to the end of this path such that their mismatch credit $\ge \mu$. Assume now $k > 0$. If $|P| = 0$, the look-up terminates and we output all the dictionary strings assigned to the current position such that their mismatch credit $\ge \mu$. Otherwise, we run a \PrefixSearch to find the longest path $\pi$ starting at $u$ that is labelled by a prefix of $P$. Suppose that $\pi$ passes through heavy paths $H_1, H_2, \dots, H_j$, leaving $H_i$ at a position $u_i$, $1 \le i \le j$. Note that for $i < j$, $u_i$ is necessarily a node of $T$, and for $i = j$ it can be a position on an edge.

Recall that for each node of $T$ we store the heavy path it belongs to, and for each heavy path we store its head. The position $u_j$ is the ending node of $\pi$. To find $u_{j-1}$, consider the heavy path $H_j$ containing $u_j$, by definition, $u_{j-1}$ is the parent of the head of $H_j$. We find all the nodes $u_i$, $1 \le i \le j$, analogously. Recall that we partitioned the dictionary strings into three types.

{\bf Strings of Type~\ref{tp:i}}. We process each path $H_i$ in turn. We select a set of nodes of the weight-balanced tree $WBT(H_i)$ covering the part of $H_i$ from the beginning and up to (but not including) $u_i$. To do this, we follow the path from the root of $WBT(H_i)$ to $u_i$ and take the nodes that hang off to the left of the path. Consider one of the selected nodes $v$ and its $(k-1)$-errata tree. All the strings in this tree have equal lengths $\delta$. To finish the recursive step, we run a dictionary look-up with $(k-1)$ mismatches for the suffix of $P$ of length $\delta$ in this tree. 

{\bf Strings of Type~\ref{tp:ii}}. We take the weight-balanced tree for $u_i$ and select a set of nodes that covers all its leaves except for the head of $H_{i+1}$. To select this set, we find the path from the root of the weight-balanced tree for $u_i$ to the head of $H_{i+1}$, and take the nodes that hang off this path. For each of the selected nodes, we run a dictionary look-up with $(k-1)$ mismatches analogously to above. We also run a dictionary look-up with $(k-1)$ mismatches with $\mu = \mu + 1$ starting from the position in $H_i$ that is one letter below $u_i$.

{\bf Strings of Type~\ref{tp:iii}}. If $u_j$ is a position on an edge, we run a dictionary look-up query with $(k-1)$ mismatches from the next position on the edge with $\mu = \mu + 1$. If $u_j$ is a node, we run two dictionary look-up queries with $(k-1)$ mismatches. First query is run in the horizontal $(k-1)$-errata tree corresponding to the set of all children of $u_j$ that are not in $H_j$. The second query is run from a position in $H_j$ that is one letter below $u_j$ with $\mu = \mu + 1$. 

Correctness of the algorithm follows from the following observation: first, we account for all dictionary strings. Second, in the case of $(k-1)$-errata trees, we account for the mismatches between the portion of the strings that we truncate and the query string via the mismatch credits. Finally, when we continue the search in the same tree, there is just one mismatch and we account for it by increasing $\mu$.

\subparagraph{Analysis.} The bounds on the space and the time complexities are summarised below. The proofs of the lemmas, which we provide in Appendix~\ref{sec:appendix} for completeness, follow closely the proofs given by Cole, Gottlieb, and Lewenstein~\cite{kerratatree}. 

\begin{restatable}{lem}{kerrataspace}\label{lm:kerrataspace}
The tries of the $k$-errata tree contain $\Oh(d \log^k d)$ strings in total. 
\end{restatable}

\begin{restatable}{lem}{kerratatime}\label{lm:kerratatime}
A dictionary look-up with $k$ mismatches for a query string $P$ requires $\Oh(\log^k d)$ operations \PrefixSearch. Apart from the time required for these operations, the algorithm spends $\Oh(\log^k d + occ)$ time.
\end{restatable}

In the next section we give an efficient implementation of \PrefixSearch under an assumption that $P$ arrives in a packed form. We will use, in particular, the following simple observation.

\begin{fact}\label{fct:prefixsearch} 
Each trie of the $k$-errata tree is built on a set of equal-length suffixes of the dictionary strings. If we run a \PrefixSearch for a suffix $S$ of the query string $P$ from a position $u$ of a trie $\tau$ of the $k$-errata tree, then the strings in the subtree of $u$ have length $|S|$. 
\end{fact}
	
\section{Prefix search for packed strings}
	\label{sec:LCP}
	We first remind several well-known data structure results that we use throughout the section. A \emph{priority queue} is a data structure like a regular queue, where each element has an integer (``priority'') associated with it. In a priority queue, an element with high priority is served before an element with low priority. A priority queue can be implemented as a heap, that for a set of $x$ elements occupies $\Oh(wx)$ bits of space and has query time $\Oh(\log x)$. 

A \emph{predecessor data structure} on a set of integers supports the following queries: Given an integer $z$, return the largest integer in the set that is at most $z$. For a set of $x$ integer keys, the predecessor data structure can be implemented as a binary search tree in $\Oh(wx)$ bits of space to support the predecessor queries in time $\Oh(\log x)$. (We do not use solutions such as~\cite{Fischer:2015,yfasttree} to avoid dependency on $m$, which will be important for our final result.) 

A \emph{dictionary data structure} stores a set of integers. A dictionary look-up receives an integer $z$ and outputs ``yes'' if $z$ belongs to the set.

\begin{lemma}[\cite{Ruzic:2008}]\label{lm:dictionary}
Let $S$ be any given set of $x$ integers. There is a dictionary over $S$ that occupies $\Oh(wx)$ bits of space, and has query time $\Oh(1)$.
\end{lemma}

We will also need \emph{lowest common ancestor queries} on tries. Given two nodes $u, v$ of a trie, their lowest common ancestor is a node of maximal depth that contains both $u$ and $v$ in its subtree. 

\begin{lemma}[\cite{LCA}]\label{lm:lca}
A trie of size $x$ can be preprocessed in $\Oh(wx)$ bits of space to maintain lowest common ancestor queries in $\Oh(1)$ time.
\end{lemma}

Finally, we need \emph{weighted level ancestor queries} on tries. A weighted level ancestor query receives a node $u$ and an integer $\ell$, and must output the deepest ancestor $u'$ of $u$ such that the length of the label of $u'$ is at most $\ell$. We will use the weighted level ancestor queries on tries for fast navigation: Suppose that we know a leaf labelled by a string $S$, then to find a position labelled by a prefix $S'$ of $S$ we can use one weighted level ancestor query instead of performing a \PrefixSearch for $S'$. To avoid dependency on $m$, we use the following simple folklore solution instead of~\cite{wLA,Amir:2007:DTS:1240233.1240242,Farach:1996:PHS:647815.738452}. 

\begin{lemma}\label{lm:wla}
A trie of size $x$ can be preprocessed in $\Oh(wx)$ bits of space to maintain weighted level ancestor queries in $\Oh(\log x)$ time.
\end{lemma}
\begin{proof}
We consider the heavy path decomposition of the trie. For each node we store a pointer to the head of the heavy path containing it, and for each path we build a binary search tree containing the length of the labels of the nodes in it. Suppose we are to answer a weighted level ancestor query for a node $u$ and an integer $\ell$. The path from the root of the trie to $u$ (which contains all the ancestors of $u$) traverses a subset of heavy paths. The size of this subset is $\Oh(\log x)$, because each time we switch paths the weight of the current node decreases by at least a factor of two. We iterate over this set of paths to find the path that contains the answer $u'$, and then use the binary search tree to find the location of $u'$ in the path. Both steps take $\Oh(\log x)$ time.
\end{proof}

\subsection{Linear space}\label{sec:LCP-linear}
As a warm-up we show a linear-space implementation of \PrefixSearch that improves the runtime of dictionary look-up queries to $\Oh(m/w + \log^{k+1} d + occ)$. Formally, we will show the following result.

\begin{theorem}\label{th:linear}
Assume a constant-size alphabet. For a dictionary $\D$ of $d > 2$ strings of length $m$, there is a data structure for dictionary look-up with $k$ mismatches that occupies $\Oh(wmd + w d \log^{k} d)$ bits of space and has query time $\Oh(m/w + \log^{k+1} d + occ)$, where $w = \Theta(\log md)$ is the size of a machine word.
\end{theorem}

Let $\Suff$ be the set of all suffixes of the strings in $\D$. We build a compact trie $T(\Suff)$ on $\Suff$. (In the literature, $T(\Suff)$ is referred to as the suffix tree of $\D$.) As the total length of the strings in $\D$ is $md$, the size of $\Suff$ is $\Oh(md)$, and therefore it occupies $\Oh(wmd)$ bits of space. We can reduce \PrefixSearch queries on the tries of the $k$-errata tree to \PrefixSearch queries on $T(\Suff)$. We distinguish between \PrefixSearch queries that start at the root of some trie of the $k$-errata tree (\emph{rooted} queries), and those that start at some inner node or even a position on an edge of a trie of the $k$-errata tree (\emph{unrooted} queries). Note that unrooted queries are used in the case $k \ge 1$ only.

\begin{restatable}{lem}{rootedLCP}\label{lm:rootedLCP}
After $\Oh(wmd + w d \log^{k} d)$ bits of space preprocessing, we can answer a rooted \PrefixSearch query for a string $Q$ and any trie of the $k$-errata tree in $\Oh(\log d)$ time given the answer to a rooted \PrefixSearch for $Q$ in $T(\Suff)$.
\end{restatable}

\begin{restatable}{lem}{unrootedLCP}\label{lm:unrootedLCP}
Assume $k \ge 1$. After $\Oh(wmd + w d \log^{k} d)$ bits of space preprocessing, we can reduce an unrooted \PrefixSearch query for $Q$ that starts at a position $u$ of a trie $\tau$ of the $k$-errata tree to a rooted \PrefixSearch for some suffix $Q'$ of $Q$ in a trie $\tau'$ of a $(k-1)$-errata tree in $\Oh(\log d)$ time given the answer to a rooted \PrefixSearch query for $Q$ in $T(\Suff)$. 
\end{restatable}

Lemmas~\ref{lm:rootedLCP} and~\ref{lm:unrootedLCP} were proved in~\cite{kerratatree}. For completeness, we give their proofs in Appendix~\ref{sec:appendix}.  Suppose we are to answer a dictionary look-up with $k$ mismatches for a string~$P$. Our algorithm traverses the $k$-errata tree and generates rooted and unrooted \PrefixSearch queries. We maintain a priority queue. Each time we need an answer to a \PrefixSearch for a string $S$ in $T(\Suff)$, we add $S$ to the priority queue. At each step of the algorithm we extract the longest string from the queue and answer the \PrefixSearch query for it. Notice that all strings in the queue are suffixes of $P$ and that the maximal length of strings in the queue cannot increase. We can therefore assume that we must answer \PrefixSearch queries for the suffixes of $P$ starting at positions $1 = i_1 \le i_2 \le \dots \le i_z$, where $z = \Oh(\log^k d)$. 

Bille, G{\o}rtz, and Skjoldjensen~\cite{packedtextindex} showed that we can preprocess  $T(\Suff)$ in linear space to answer \PrefixSearch queries for a single query string of length $m$ in $\Oh(m/w+\log \log md)$ time. As an immediate corollary we obtain that we can answer $z$ \PrefixSearch queries in $z \cdot \Oh(m/w+\log \log md)$ time, but this is too slow for our purposes. Below we develop their ideas to give a more efficient approach.

\begin{lemma}~\label{lm:packedtrie}
$T(\Suff)$ can be preprocessed in $\Oh(wmd)$ bits of space to answer \PrefixSearch for the suffixes of $P$ starting at positions $1 = i_1 \le i_2 \le \dots \le i_z$ in $\Oh(m/w + z \log d)$ time.
\end{lemma}
\begin{proof}
We assume that the strings in the dictionary are stored in the packed form. By construction, each edge of a trie of the $k$-errata is labelled by a substring of a dictionary string. It means that we can store each label as three integers: the id of the string, and the starting and the ending positions of the substring. Next, we preprocess $T(\Suff)$ for weighted level ancestor queries (Lemma~\ref{lm:wla}). A node or a position in the trie is called \emph{boundary} if the length of its label is a multiple of $w/\log \sigma$, where $w$ is the size of a machine word and $\sigma$ is the size of the alphabet. Boundary nodes cut the tree into micro-trees. We only consider the micro-trees containing more than two nodes. We define the label of a leaf of a micro-tree as a machine word that contains a packed representation of the string written on the path from the root of the micro-tree to the leaf. The labels can be treated as integers; for each micro-tree we create a dictionary (Lemma~\ref{lm:dictionary}) and a predecessor data structure on the labels of its leaves. We also preprocess each micro-tree for lowest common ancestors. Note that the total size of the micro-trees is $\Oh(md)$, as each edge of $T(\Suff)$ contains at most two nodes of the micro-trees. Therefore, the preprocessing requires $\Oh(wmd)$ bits of space.

We now explain how to answer the \PrefixSearch queries for the suffixes of $P$ starting at the positions $1= i_1 \le i_2 \le \dots \le i_z$. For $i_1$, we start at the root of $T(\Suff)$. For $i_j$, $j > 1$, we use the information obtained at the previous step. Namely, suppose that the \PrefixSearch for $P[i_{j-1},m]$ terminated at a position labelled by $P[i_{j-1},\ell_{j-1}]$. We take any leaf below this position, let it be labelled by a string $S \in \Suff$. Let $P[i_j, \ell'_{j-1}]$ be the longest prefix of $P[i_{j},\ell_{j-1}]$ such that its length is a multiple of $w/\log \sigma$. We then start the \PrefixSearch from a position $u$ labelled by $P[i_j, \ell'_{j-1}]$. To find the position $u$, we first find the leaf labelled by $S[i_j-i_{j-1}+1,|S|] \in \Suff$, and then use a weighted level ancestor query to jump to $u$ in $\Oh(\log d)$ time. Notice that $u$ is boundary. If $u$ is not a root of a micro-tree, it has a single outgoing edge of length at least $w/\log\sigma$. We compare the first $w/\log \sigma$ letters of the label of this edge and $P[\ell'_{j-1}+1, \ell'_{j-1}+w/\log \sigma]$ in $\Oh(1)$ time by comparing the corresponding machine words. If they are equal, we continue from the next boundary node on the edge in a similar manner. Otherwise, we find the first mismatch between the two strings in $\Oh(1)$ time as follows: First, compute a bitwise XOR of the two strings, and then locate the most significant bit using the technique of~\cite{bitwise}. 

If $u$ is the root of a micro-tree $\tau$, we search for $P[\ell'_{j-1}+1, \ell'_{j-1}+w/\log \sigma]$ in the dictionary of~$\tau$. If it is in the dictionary and corresponds to a leaf $v$, we continue to $v$. Otherwise, we find its predecessor $\pred$ and successor $\successor$ using the predecessor data structure. The \PrefixSearch must terminate either on the path from $u$ to the leaf of the micro-tree labelled by $\pred$, or on the path from $u$ to the leaf of the micro-tree labelled by $\successor$. We compute the longest common prefix of $P[\ell'_{j-1}+1, \ell'_{j-1}+w/\log \sigma]$ with $\pred$ and with $\successor$ using bitvector operations in $\Oh(1)$ time as explained above, take the longest of the two, and find the position labelled by it in $\Oh(\log d)$ time using a weighted level ancestor query.  

The running time of each prefix search query is proportional to the number of $(w/\log\sigma)$-length blocks of $P$ that we compare with the labels of the edges of $T(\Suff)$. Notice that each two different \PrefixSearch queries share at most one block of letters. Therefore, as the size of the alphabet $\sigma$ is constant, the total running time of $z$ \PrefixSearch queries is $\Oh(m/w + z\log d)$. 
\end{proof}

Lemmas~\ref{lm:kerrataspace},~\ref{lm:kerratatime},~\ref{lm:rootedLCP},~\ref{lm:unrootedLCP}, and~\ref{lm:packedtrie} give Theorem~\ref{th:linear}.

\subsection{Entropy-bounded space}\label{sec:LCP-entropy}
In this section we improve the space requirements of our implementation of \PrefixSearch and show the following theorem.


\begin{theorem}\label{th:entropy}
Assume a constant-size alphabet. For a dictionary of $d > 2$ strings of lengths $m$ and any $q = o(\log md)$, let $H_q$ be the $q$-th empirical entropy of the concatenation of all the dictionary strings. There is a data structure for dictionary look-ups with $k$ mismatches that uses $md H_q + o(md) + \Oh(w d \log^{k} d)$ bits of space and has query time $\Oh(m/w + \log^{k+1} d + occ)$, where $w = \Theta(\log md)$ is the size of a machine word.
\end{theorem}

There are two bottlenecks: First, we need to store the dictionary strings, and second, the tree structure of $T(\Suff)$ requires $\Omega(md)$ space. To overcome the first bottleneck, we replace the packed representation of the dictionary strings by the Ferragina-Venturini representation:

\begin{lemma}[\cite{entropy}]\label{lm:entropy_scheme}
Under the assumption of an alphabet of constant size $\sigma$, for any $q = o(\log md)$ there exists a data structure that uses $md H_q + o(md)$ bits of space and supports constant-time access to any $w/\log \sigma = \Theta(\log md)$-length substring of a dictionary string. 
\end{lemma}

If $d > 2$ is a constant, if suffices to store the Ferragina-Manzini representation of the dictionary strings to obtain the bounds of Theorem~\ref{th:entropy}. Indeed, when a query string $P$ arrives, we can decide if the Hamming distance between $P$ and a dictionary string is at most $k$ in $\Oh(m/w+k) = \Oh(m/w + \log^k d)$ time, using comparison by machine words and bitvector operations. As $d$ is constant, we obtain the desired time bound. Below we assume that $d = \Omega(1)$.

We now deal with the second bottleneck. We will consider a smaller trie $T(\Suff')$ on a subset $\Suff'$ of $\Suff$, and will show that \PrefixSearch queries on tries of the $k$-errata tree can be reduced to \PrefixSearch queries on this trie. $\Suff'$ is defined to be the set of all suffixes of the dictionary strings that start at positions $w \log d / \log \sigma$, $2 w \log d / \log \sigma$, and so on. We call such suffixes \emph{sampled}. Below we show that we can reduce \PrefixSearch queries in the tries of the $k$-errata tree to \PrefixSearch queries in $T(\Suff')$.

\begin{lemma}\label{lm:rootedLCP-2}
After $\Oh(md / \log d + w d \log^{k} d)$ bits of space preprocessing, we can answer a rooted \PrefixSearch query for a string $Q = P[\ell,m]$ in $\Oh(\log d)$ time given the answer to a rooted \PrefixSearch query for a string $P[\ell',m]$ in $T(\Suff')$, where $\ell' \ge \ell$ is the smallest multiple of $w \log d / \log \sigma$. 
\end{lemma}
\begin{proof}
At the preprocessing step, we traverse $T(\Suff')$ and remember the leftmost and the rightmost leaves in each of its subtrees. We also remember the neighbours of each leaf in the left-to-right order, and finally we preprocess the trie for lowest common ancestor queries. As a second step we preprocess each trie of the $k$-errata tree for lowest common ancestor and weighted level ancestor queries. We also build the following data structure for each trie of the $k$-errata tree. For each string $S$ in the trie, let $S = p S'$, where $S'$ is the longest sampled suffix of $S$. We call $p$ a head of $S$, and define the rank of $S$ to be the rank of $S'$ in $\Suff'$. We build a compact trie $T_{heads}$ containing the heads of all the strings, and preprocess it as in Lemma~\ref{lm:packedtrie}. If $T_{heads}$ contains $x$ strings, we use $\Oh(wx)$ bits of space for the preprocessing, i.e. $\Oh(w d \log^k d)$ bits of space in total. We also associate a predecessor data structure with each of its leaves. The predecessor data structure of a leaf labelled by $p$ contains the ranks of all the strings such that their head is equal to $p$. The predecessor data structures occupy $\Oh(w d \log^k d)$ bits of space in total as well.

Suppose we are to answer a rooted \PrefixSearch query for a string $Q = P[\ell,m]$ and a trie $\tau$ of the $k$-errata tree. Let $T_{heads}$ be the compact trie containing the heads of the strings in $\tau$. By Fact~\ref{fct:prefixsearch}, the length of the heads is $(\ell'-\ell)$.  We first read $P[\ell,\ell'-1]$ in blocks of $w / \log \sigma$ letters in $\Oh(\log d)$ time, and run a \PrefixSearch for it in $T_{heads}$ in $\Oh(\log d)$ time. If the \PrefixSearch terminates in a position $u$ of $T_{heads}$ that is not in a leaf, it remains to find the position corresponding to $u$ in $\tau$, which we can do with one weighted level ancestor query.

Assume now that the \PrefixSearch terminates in a leaf of $T_{heads}$. By the condition of the lemma, we know the answer to the rooted \PrefixSearch for $P[\ell',m]$ in $T(\Suff')$. We also store the leftmost and the rightmost leaves in each subtree of $T(\Suff')$, and therefore can find the predecessor of $P[\ell',m]$ in $\Suff'$ in $\Oh(1)$ time. We use the predecessor data structure associated with the leaf to find the predecessor $\pred$ and successor $\successor$ of $P[\ell',m]$ in $\Oh(\log d)$ time. To find the position where the \PrefixSearch for $P$ terminates, we compute the lengths $\ell_p, \ell_s$ of the longest common prefix of $P[\ell',m]$ and $\pred$ and of $P[\ell',m]$ and $\successor$. We can compute the longest common prefix of $P[\ell',m]$ and $\pred$ (which is a sampled suffix of a dictionary string) in $\Oh(1)$ time via a lowest common ancestor query on $T(\Suff')$. We then compute $\ell_s$ in a similar way.  If $\ell_p = \ell_s$, we return the lowest common ancestor of $\pred$ and $\successor$ as the answer. If $\ell_p > \ell_s$, then the answer is the ancestor of $P[\ell,\ell'-1] \circ \pred$ such that the length of its label is $(\ell'-\ell)+\ell_p$, and we can find it by one weighted level ancestor query. The case $\ell_s > \ell_p$ is analogous.
\end{proof}

\begin{lemma}\label{lm:unrootedLCP-2}
Assume $k \ge 1$. After $\Oh(md / \log d + w d \log^{k} d)$ bits of space preprocessing, we can answer an unrooted \PrefixSearch query for a string $Q = P[\ell,m]$ by reducing it to a rooted \PrefixSearch query in $\Oh(\log d)$ time given the answer to a rooted \PrefixSearch query for a string $P[\ell',m]$ in $T(\Suff')$, where $\ell' \ge \ell$ is the smallest multiple of $w \log d / \log \sigma$.
\end{lemma}
\begin{proof}
During the preprocessing step, we preprocess $T(\Suff')$ for lowest common ancestor queries and each trie of the $k$-errata tree for weighted level ancestor queries. Let $u$ be the position in a trie $\tau$ where we start the \PrefixSearch for $P[\ell,m]$. The search path for $P[\ell,m]$ traverses a number of heavy paths. The first path is the path containing $u$. Let $S$ be the label of the part of the path starting from $u$. We consider two cases. Suppose first that $\ell' = \ell$. In this case, $S$ is a suffix of one of the dictionary strings starting at a position $\ell'$, i.e. it is sampled. Therefore, we can find the longest common prefix of $P[\ell',m]$ and $S$ using one lowest common ancestor query on $T(\Suff')$. We can then find the node in the path corresponding to this longest common prefix using one weighted level ancestor query. From there, we can find the starting node of the second heavy path traversed by $P[\ell,m]$ in $\Oh(1)$ time. It remains to answer a rooted \PrefixSearch query in the subtree rooted at this node, which is a trie of a $(k-1)$-errata tree by construction. When we know the answer for this \PrefixSearch, we can go back to $\tau$ using one weighted level ancestor query. In the second case $\ell' > \ell$. We start by comparing $P[\ell+1,m]$ and $S$ by blocks of $w/\log \sigma$ letters until we reach the start of a sampled suffix, and then proceed as above. 
\end{proof}

Suppose that we are to answer a dictionary look-up with $k$ mismatches for a string $P$. Our algorithm traverses the $k$-errata tree and generates rooted \PrefixSearch queries for the suffixes of $P$ in $T(\Suff')$. We maintain a priority queue. Each time we need an answer to a rooted \PrefixSearch for a suffix $Q = P[i,m]$ in $T(\Suff')$, we add $Q$ to the priority queue. At each step we extract the longest string from the queue and answer the \PrefixSearch query for it. Since the maximal length of suffixes in the queue cannot increase, we can assume that we must answer \PrefixSearch queries for the suffixes of $P$ starting at positions $i_1 \le i_2 \le \dots \le i_z$, where $z = \Oh(\log^k d)$. Moreover, for each $j$ the position $i_j$ is a multiple of $w \log d / \log \sigma$. We preprocess $T(\Suff')$ as in Lemma~\ref{lm:packedtrie}, which requires $\Oh(md /\log d) = o(md)$ bits of space. We first run \PrefixSearch for $P[i_1,m]$ in $\Oh(m/w)$ time. Suppose it follows the path labelled by $P[i_1, \ell_1]$. Let $S \in \Suff'$ be an arbitrary string that ends below the end of this path. We then find the leaf corresponding to $S[i_2-i_1,|S|]$. By construction of $T(\Suff')$ and because $i_2-i_1$ is a multiple of $w \log d / \log \sigma$, such a leaf must exist. We then use a weighted level ancestor query to find the end of the path labelled by $P[i_2, \ell'_1]$, where $P[i_2, \ell'_1]$ is the longest suffix of $P[i_2, \ell_1]$ such that its length is a multiple of $w / \log \sigma$, and continue the \PrefixSearch for $P[i_2,m]$ from there, and so on. The total running time is $\Oh(m/w + z\log d) = \Oh(m/w + \log^{k+1} d)$. 

If $d = \Omega(1)$ as we assumed earlier, lemmas~\ref{lm:entropy_scheme},~\ref{lm:rootedLCP-2},~\ref{lm:unrootedLCP-2}, and the discussion above give Theorem~\ref{th:entropy}.
	
\section{Removing extra logarithm from the time complexity}
	\label{sec:removeloglog}
	In this section we improve the query time to $\Oh(m/w + \log^k d + occ)$ and show our final result.

\begin{theorem}\label{th:main}
Assume a constant-size alphabet. For a dictionary of $d > 2$ strings of lengths $m$ and for any $q = o(\log md)$, let $H_q$ be the $q$-th empirical entropy of the concatenation of all strings in the dictionary. There exists a data structure for dictionary look-ups with $k$ mismatches that uses $2 md H_q + o(md) +  \Oh(wd \log^{k} d)$ bits of space and has query time $\Oh(m/w + \log^k d + occ)$, where $w = \Theta(\log md)$ is the size of a machine word.
\end{theorem}

As explained in Theorem~\ref{th:entropy}, we can assume $d = \Omega(1)$. Recall that the dictionary look-up with $k$ mismatches is run recursively. The first $(k-2)$ levels of recursion require $\Oh(\log^{k-2} d)$ \PrefixSearch queries and can be implemented in $\Oh(m/w+\log^{k-1} d)$ time. Therefore, it suffices to improve the runtime of the two last levels of the recursion, where we must perform a batch $\Oh(\log^{k-1} d)$ dictionary look-up queries with one mismatch. To achieve the desired complexity we use the fact that the queries are related, as explained below.

\subparagraph{Preprocessing.} For a string $S = s_1 s_2 \dots s_m$ we define its reverse $S^R = s_m \dots s_2 s_1$. First, we build a compact trie on the reverses of all the dictionary strings and preprocess it as described in Lemma~\ref{lm:packedtrie}, which takes $\Oh(w d)$ bits of space. We store the reverses using the Ferragina-Venturini representation (Lemma~\ref{lm:entropy_scheme}) in $H'_q md + o(md)$ bits of space, where $H'_q$ is the $q$-th empirical entropy of the reverse of the concatenation of all the strings in the dictionary. By~\cite[Theorem A.3]{Ferragina:2005:ICT:1082036.1082039}, $H'_q md + o(md) = H_q md + o(md)$. For the second step, we need Karp-Rabin fingerprints. We modify the standard definition as we work with packed strings.

\begin{definition}[Karp-Rabin fingerprints~\cite{KarpRabin}]
Consider a string $S$ and its packed representation $ w_1 w_2 \dots w_z$, where each $w_i$ is a machine word. (If $|S|$ is not a multiple of $w/\log\sigma$, we append an appropriate number of zeros.) The Karp–Rabin fingerprint of $S$ is defined as $\varphi = \sum_{i = 1}^z w_i \cdot r^{z-i} \bmod p$, where $p$ is a fixed prime number and $r$ is a randomly chosen integer in $[0,p-1]$.
\end{definition}

From the definition it follows that if the strings are equal, their fingerprints are equal. Furthermore, it is well-known that for any $c > 3$ and $p > (\max\{m/w, d \log^k d\})^c$, the probability of two distinct strings of length $z w \le \max\{m, w d \log^k d\}$ having the same fingerprint (\emph{collision probability}) is less than $1/(\max\{m/w, d \log^k d\})^{c-1}$. Consider a trie $\tau$ of the $k$-errata tree. By definition, the lengths of the leaf labels in $\tau$ is at most $m$. From the bound on the collision probability it follows that we can choose $p$ and $r$ so that the fingerprints of the reverses of these labels are distinct. For each leaf of $\tau$, we compute the Karp-Rabin fingerprint of the reverse of its label and add it to a dictionary (Lemma~\ref{lm:dictionary}) associated with $\tau$. Also, using the same $p$ and $r$, we compute Karp-Rabin fingerprints corresponding to inner nodes of the tries of the $k$-errata tree. Namely, consider one of such nodes, and let $S$ be its label and $\delta$ be the length of the strings in the trie. We take the reverse of $S$, prepend it with $(\delta-|S|) \bmod w/\log\sigma$ zeros, and compute the Karp-Rabin fingerprint of the resulting string.

\subparagraph{Queries.} 
We must run $\Oh(\log^{k-1} d)$ dictionary look-up queries with one mismatch. Consider one of these queries, let it be a query for a string $Q$ (which must be a suffix of $P$) in a trie $\tau$ and recall the algorithm of Section~\ref{sec:k-errata}. First, we run a \PrefixSearch to find the longest path $\pi$ that is labelled by a prefix of $Q$. For this step we can use $T(\Suff')$, as the total number of such queries is $\Oh(\log^{k-1} d)$ and therefore we can spend $\Oh(\log d)$ time per each of them. Suppose that $\pi$ traverses the heavy paths $H_1, H_2, \dots, H_j$ and leaves the heavy path $H_i$ at a position $u_i$. We can find the positions $u_j$ in $\Oh(\log d)$ time once we have found the end of $\pi$. The rest of the algorithm can be described as follows. First, we must perform dictionary look-ups with $0$ mismatches (i.e., \PrefixSearch) in $\Oh(\log d)$ vertical and $\Oh(\log d)$ horizontal $0$-errata trees (that are tries of the $k$-errata tree by definition). Second, for each $1 \le i < j$, we must perform a dictionary look-up with $0$ mismatches (\PrefixSearch) from a position $u'_i$ that follows $u_i$ in the heavy path $H_i$. Importantly, each $u_i$ is a node. Finally, we must perform a a dictionary look-up with $0$ mismatches (\PrefixSearch) from a position $u'_j$ that follows $u_j$ in the heavy path $H_j$.

We note that to perform the \PrefixSearch from the position $u'_j$ we can use $T(\Suff')$, as before, because the total number of such \PrefixSearch operations is $\Oh(\log^{k-1} d)$. We now explain how we perform the \PrefixSearch operations in vertical and horizontal $0$-errata trees, as well as the \PrefixSearch operations from nodes $u'_i$, $1 \le i < j$.  In total, we must perform $\Oh(\log^k d)$ such operations, and for each of them the query string is a suffix of $P$. Let $P[i_1,m], P[i_2,m], \ldots, P[i_z,m]$, $z = \Oh(\log^k d)$ be the suffixes of $P$ for which we are to run a \PrefixSearch. We create a bitvector of length $m = o(md)$ where each $i_j^{th}$ bit is set. We then compute the Karp-Rabin fingerprints of the reverses of  $P[i_1,m], P[i_2,m], \ldots, P[i_z,m]$ in $\Oh(m/w + z)$ time using the following fact. 

\begin{fact}\label{fact:fingerprint}
Given the Karp-Rabin fingerprints of $X$ and $Y$, where the length of $X$ is a multiple of $w/\log \sigma$, we can compute the Karp-Rabin fingerprint of their concatenation, $XY$ in $\Oh(1)$ time.
\end{fact}

We iterate over all blocks of the bitvector starting from the last one and maintain the Karp-Rabin fingerprint of the reverse of the suffix of $P$ that starts at the current position. When we start a new block, we update the Karp-Rabin fingerprint. If a block contains set bits (which we can decide in constant time), we extract the positions of all set bits in $\Oh(1)$ time per bit using the technique of~\cite{bitwise}, and compute the corresponding Karp-Rabin fingerprints. Also, as a preliminary step, we run a \PrefixSearch for $P^R$ in the compact trie on the reverses of the dictionary strings in $\Oh(m/w + \log d)$ time. Let $u$ be the position where this \PrefixSearch terminates. 

\PrefixSearch in vertical and horizontal $0$-errata trees. Assume we must answer a \PrefixSearch for $P[i_j,m]$ on a tree $\tau$. We search the fingerprint of the reverse of $P[i_j,m]$ in the dictionary associated with $\tau$. The search will return at most one leaf of the tree. We know that its label is equal to $P[i_j,m]$ with high probability, but we need a deterministic answer. We test the leaf as follows.  Let $S$ be one of the dictionary strings such that its id is stored at the leaf. We find the leaf $v$ of the compact tree on the reverses of the dictionary strings that corresponds to the reverse $S^R$ of $S$. Now, we can compute the length of the longest common prefix of $P^R$ and $S^R$ in constant time via a lowest common ancestor query for $u$ and $v$ and check if it is indeed equal or larger than $|P[i_j,m]|$. 

\PrefixSearch from $u'_i$, $1 \le i < j$. This step is equivalent to the following: Find all the strings in the trie that start with a label of $u'_i$ and end with a given suffix of $P$. We can compute the Karp-Rabin fingerprints of the reverses of the strings that we are looking for as follows. Positions $u_i$ are necessarily nodes and we store the Karp-Rabin fingerprints of the reverses of their labels. Recall that if $S_i$ was the label of $u_i$, we prepended the reverse $S_i^R$ of $S_i$ with $(\delta-|S|) \bmod w/\log \sigma$ zeros, where $\delta$ is the length of the strings in the trie containing $u_i$. It follows that we can compute the fingerprint $\varphi_i$ of the reverse of the label of $u'_i$ prepended with $(\delta-|S_i|-1) \bmod w/\log\sigma$ zeros in $\Oh(1)$ time. Knowing $\varphi_i$ and the fingerprint of the reverse of the suffix of $P$, we can compute the fingerprint of the strings we are searching for in constant time. We note that prepending with zeros is necessary in order to align the borders of the blocks in the reverse of the label of $u'_i$ and the reverse of the label of the suffix of $P$. We finish the computation as above, that is we find a leaf such that the fingerprint of the reverse of its label is equal to the fingerprint of the strings we are looking for, and test it using the trie on the reverses of the dictionary strings.

\bibliographystyle{plainurl}
\bibliography{main}

\begin{thebibliography}{10}

\bibitem{Amir:2007:DTS:1240233.1240242}
Amihood Amir, Gad~M. Landau, Moshe Lewenstein, and Dina Sokol.
\newblock Dynamic text and static pattern matching.
\newblock {\em ACM Trans. Algorithms}, 3(2), May 2007.

\bibitem{Andoni:2015:ODH:2746539.2746553}
Alexandr Andoni and Ilya Razenshteyn.
\newblock Optimal data-dependent hashing for approximate near neighbors.
\newblock In {\em Proc. of the Forty-seventh Annual ACM Symposium on Theory of
  Computing}, STOC'15, pages 793--801, 2015.

\bibitem{Belazzougui2009FasterAS}
Djamal Belazzougui.
\newblock Faster and space-optimal edit distance ``1'' dictionary.
\newblock In {\em Proc. of the Annual Symposium on Combinatorial Pattern
  Matching}, CPM'09, pages 154--167, 2009.

\bibitem{Belazzougui2012CompressedSD}
Djamal Belazzougui and Rossano Venturini.
\newblock Compressed string dictionary look-up with edit distance one.
\newblock In {\em Proc. of the Annual Symposium on Combinatorial Pattern
  Matching}, CPM'12, pages 280--292, 2012.

\bibitem{packedtextindex}
Philip Bille, Inge~Li G{\o}rtz, and Frederik~Rye Skjoldjensen.
\newblock Deterministic indexing for packed strings.
\newblock In {\em Proc. of the Annual Symposium on Combinatorial Pattern
  Matching}, CPM'17, pages 6:1--6:11, 2017.

\bibitem{BHSH2007}
Thomas Bocek, Ela Hunt, Burkhard Stiller, and Fabio Hecht.
\newblock Fast similarity search in large dictionaries.
\newblock Technical Report ifi-2007.02, Department of Informatics, University
  of Zurich, 2007.

\bibitem{Brodal1996ApproximateDQ}
Gerth~St{\o}lting Brodal and Leszek Gasieniec.
\newblock Approximate dictionary queries.
\newblock In {\em Proc. of the Annual Symposium on Combinatorial Pattern
  Matching}, CPM'96, pages 65--74, 1996.

\bibitem{Brodal2000ImprovedBF}
Gerth~St{\o}lting Brodal and Srinivasan Venkatesh.
\newblock Improved bounds for dictionary look-up with one error.
\newblock {\em Inf. Process. Lett.}, 75:57--59, 2000.

\bibitem{Chan2006CompressedIF}
Ho-Leung Chan, Tak~Wah Lam, Wing-Kin Sung, Siu-Lung Tam, and Swee-Seong Wong.
\newblock Compressed indexes for approximate string matching.
\newblock {\em J. Algorithmica}, 58:263--281, 2006.

\bibitem{Chan2006ALS}
Ho-Leung Chan, Tak~Wah Lam, Wing-Kin Sung, Siu-Lung Tam, and Swee-Seong Wong.
\newblock A linear size index for approximate pattern matching.
\newblock In {\em Proc. of the Annual Symposium on Combinatorial Pattern
  Matching}, CPM'06, pages 45--59, 2006.

\bibitem{CL2015}
Timothy Chan and Moshe Lewenstein.
\newblock Fast string dictionary lookup with one error.
\newblock In {\em Proc. of the Annual Symposium on Combinatorial Pattern
  Matching}, CPM'15, pages 114--123, 2015.

\bibitem{CG2017}
Aleksander Cis\l{}ak and Szymon Grabowski.
\newblock A practical index for approximate dictionary matching with few
  mismatches.
\newblock {\em Computing \& Informatics}, 36(5):1088--1106, 2017.

\bibitem{kerratatree}
Richard Cole, Lee-Ad Gottlieb, and Moshe Lewenstein.
\newblock Dictionary matching and indexing with errors and don't cares.
\newblock In {\em Proc. of the 36th Annual ACM Symposium on Theory of
  Computing}, STOC'04, pages 91--100, 2004.

\bibitem{Farach:1996:PHS:647815.738452}
Martin Farach and S.~Muthukrishnan.
\newblock Perfect hashing for strings: Formalization and algorithms.
\newblock In {\em Proc. of the Annual Symposium on Combinatorial Pattern
  Matching}, CPM'96, pages 130--140, 1996.

\bibitem{Ferragina:2005:ICT:1082036.1082039}
Paolo Ferragina and Giovanni Manzini.
\newblock Indexing compressed text.
\newblock {\em J. ACM}, 52(4):552--581, July 2005.

\bibitem{entropy}
Paolo Ferragina and Rossano Venturini.
\newblock A simple storage scheme for strings achieving entropy bounds.
\newblock {\em Theoretical Computer Science}, 372(1):115 -- 121, 2007.

\bibitem{Fischer:2015}
Johannes Fischer and Pawel Gawrychowski.
\newblock Alphabet-dependent string searching with wexponential search trees.
\newblock In {\em Proc. of the Annual Symposium on Combinatorial Pattern
  Matching}, CPM'15, pages 160--171, 2015.

\bibitem{LCA}
Johannes Fischer and Volker Heun.
\newblock Theoretical and practical improvements on the {RMQ}-problem, with
  applications to {LCA} and {LCE}.
\newblock In {\em Proc. of the Annual Conference on Combinatorial Pattern
  Matching}, CPM'06, pages 36--48, 2006.

\bibitem{bitwise}
Michael~L. Fredman and Dan~E. Willard.
\newblock Surpassing the information theoretic bound with fusion trees.
\newblock {\em J. Comput. Syst. Sci.}, 47(3):424--436, December 1993.

\bibitem{wLA}
Pawel Gawrychowski, Moshe Lewenstein, and Patrick~K. Nicholson.
\newblock Weighted ancestors in suffix trees.
\newblock In {\em Proc. of the Annual European Symposium on Algorithms},
  ESA'14, pages 455--466, 2014.

\bibitem{heavypath}
Dov Harel and Robert~Endre Tarjan.
\newblock Fast algorithms for finding nearest common ancestors.
\newblock {\em SIAM Journal on Computing}, 13(2):338--355, 1984.

\bibitem{Hon2011CompressedDM}
Wing-Kai Hon, Tsung-Han Ku, Rahul Shah, Sharma~V. Thankachan, and Jeffrey~Scott
  Vitter.
\newblock Compressed dictionary matching with one error.
\newblock In {\em Proc. of the Data Compression Conference}, DCC'11, pages
  113--122, 2011.

\bibitem{Karch2010ImprovedFS}
Daniel Karch, Dennis Luxen, and Peter Sanders.
\newblock Improved fast similarity search in dictionaries.
\newblock In {\em Proc. of the International Symposium on String Processing and
  Information Retrieval}, SPIRE'10, pages 173--178, 2010.

\bibitem{KarpRabin}
Richard~M. Karp and Michael~O. Rabin.
\newblock Efficient randomized pattern-matching algorithms.
\newblock {\em IBM J. Res. Dev.}, 31(2):249--260, March 1987.

\bibitem{Lam2005ImprovedAS}
Tak~Wah Lam, Wing-Kin Sung, and Swee-Seong Wong.
\newblock Improved approximate string matching using compressed suffix data
  structures.
\newblock {\em J. Algorithmica}, 51:298--314, 2005.

\bibitem{Manzini:2001:ABT:382780.382782}
Giovanni Manzini.
\newblock An analysis of the {B}urrows~-{W}heeler transform.
\newblock {\em J. ACM}, 48(3):407--430, May 2001.

\bibitem{Mor:1982:HCM:358728.358752}
Moshe Mor and Aviezri~S. Fraenkel.
\newblock A hash code method for detecting and correcting spelling errors.
\newblock {\em Commun. ACM}, 25(12):935--938, December 1982.

\bibitem{Ruzic:2008}
Milan Ru\v{z}i\'{c}.
\newblock Uniform deterministic dictionaries.
\newblock {\em ACM Trans. Algorithms}, 4(1):1:1--1:23, March 2008.

\bibitem{packedtrie}
Takuya Takagi, Shunsuke Inenaga, Kunihiko Sadakane, and Hiroki Arimura.
\newblock Packed compact tries: A fast and efficient data structure for online
  string processing.
\newblock In {\em Proc. of the 27th International Workshop on Combinatorial
  Algorithms}, volume 9843 of {\em IWOCA'16}, pages 213--225. Springer, 2016.

\bibitem{yfasttree}
Dan~E. Willard.
\newblock Log-logarithmic worst-case range queries are possible in space o(n).
\newblock {\em Information Processing Letters}, 17(2):81 -- 84, 1983.

\bibitem{Yao1997DictionaryLW}
Andrew Chi-Chih Yao and Foong~Frances Yao.
\newblock Dictionary look-up with one error.
\newblock {\em J. Algorithms}, 25:194--202, 1997.

\end{thebibliography}

\appendix
\section{Missing proofs}\label{sec:appendix}
Let the weight of a node $u$ of a weight-balanced tree be the sum of the weights of the leaves in the subtree rooted at $u$.

\begin{fact}\label{fact:wbt}
The weights of the nodes in any root-to-leaf path of a weight-balanced tree decrease. Moreover, if $(u, v)$ is an edge in the path, and $v$ is not a leaf, then the weight of $v$ is at least two times smaller than the weight of $u$. 
\end{fact}

The proofs of the following lemmas follow closely the proofs given by Cole, Gottlieb, and Lewenstein~\cite{kerratatree}, and we provide them only for completeness. We note that the upper bounds in~\cite{kerratatree} are more precise, but the ones below suffice for our purposes.

\kerrataspace*
\begin{proof}
Let $S_k(d)$ denote the number of strings in the tries of the $k$-errata tree for a dictionary of size $d$. We will show by induction on $k$ that $S_k(d)$ can be upper-bounded by 

$$\tilde{S}_k(d') =  2 \cdot 4^{k} d'  {\log d' + k \choose \log d'} - d'$$
where $d'$ is the nearest power of two larger than $d$. The right-hand side of the inequality is $\Oh(d \log^k d)$, which will give the claim of the lemma. Indeed, using the inequality ${n \choose k} \leq (en/k)^{k}$ we can upper bound the right-hand side by $c^k d (\log d + k)^k / k^k$ for some constant $c$. For $d$ large enough, $\log d \ge 2c$ so $c^k d (\log d + k)^k / k^k \le d \log^{k} d \cdot  (c/k + 1/2)^{k} = \Oh(d\log^{k} d)$.

For $k = 0$ we have $S_k(d) = d$ as the $k$-errata tree contains only one trie, so the base case holds. Let now $k \geq 1$. Recall that we start the construction of the $k$-errata tree by creating a compact trie on $d$ strings, and then build the $(k-1)$-errata trees recursively. We now estimate the number of strings in the tries of the $(k-1)$-errata trees. For that, we will consider each string $S$ in the dictionary, and will estimate its contribution to the total size of the trees. Let $S$ correspond to a leaf of the trie $T$ such that the path from the root to this leaf passes through heavy paths $H_1, H_2, \dots, H_j$ leaving a heavy path $H_i$ at a node $u_i$. 

For each heavy path we build a number of vertical $(k-1)$-errata trees that contain~$S$. The claim is that there are at most two trees containing at most $d$ strings, at most two trees containing at most $d/2$ strings, etc. Indeed, from the definition of heavy paths it follows that the total weight of the nodes in a heavy path $H_{i+1}$ is at least two times smaller than the total weight of the nodes in a heavy path $H_{i}$. Furthermore, consider the sequence of the vertical $(k-1)$-errata trees containing $S$ in the top-down order. By Fact~\ref{fact:wbt} we obtain that the size of the trees in this sequence must decrease by a factor of at least two each time except for when we arrive to a leaf of a weight-balanced tree corresponding to a heavy path $H_i$, where it can decrease by one. Therefore, each two trees the size decreases by at least a factor of two. The claim follows. 

A similar claim holds for the horizontal $k$-errata trees: among the trees containing $S$, there are at most two trees of size at most $d$, at most two trees of size at most $d/2$, etc. The string $S$ belongs to the horizontal $(k-1)$-errata trees associated with the nodes $u_1, u_2, \dots, u_{j-1}$ only. The total size of the horizontal $(k-1)$-errata trees for $u_i$ decreases by a factor of at least two each time we switch paths (recall that we build these trees for all the children of $u_i$ but the heavy one). Consider the sequence of the horizontal $(k-1)$-errata trees containing $S$ in the top-down order. By Fact~\ref{fact:wbt} we obtain that the size of the trees in this sequence must decrease by a factor of at least two each time except for when we arrive to a leaf of a weight-balanced tree corresponding to a heavy path $H_i$, where it can decrease by one. Therefore, each two trees the size decreases by at least a factor of two, and we obtain the claim as before.

Let $d'$ be the nearest power of two larger than $d$. By the induction hypothesis and because $\tilde{S}_{k-1}(x)/x$ is a non-decreasing function of~$x$, we have

$$S_k(d) \le 4 d' \left[\frac{\tilde{S}_{k-1}(d')}{d'} + \frac{\tilde{S}_{k-1}(d'/2)}{d'/2} + \dots + \frac{\tilde{S}_{k-1}(1)}{1}\right] + d'$$
Plugging in the expression for $\tilde{S}_{k-1}$ and simplifying the sums of binomial coefficients (we use the fact that $\log d'$ is an integer),  we obtain

\begin{align*}
S_k(d) &\le 2 \cdot 4^k d' \big({\log d' + k - 1 \choose \log d'} + {\log d' + k - 2 \choose \log d' - 1} + \ldots + 1 \big) - 4 d' (1+\log d') + d'  \\
&= 2 \cdot 4^k d' {\log d' + k \choose \log d'} - 4 d' (1+\log d') + d'  \le 2 \cdot 4^k d' {\log d' + k \choose \log d'} - d' = \tilde{S}_k(d') \qedhere
\end{align*}
\end{proof}

\kerratatime*
\begin{proof}
Let $T_k(d)$ be the number of \PrefixSearch operations that we run while performing a dictionary look-up with $k$ mismatches for a dictionary of size $d$. We will show by induction on $k$ that $T_k(d)$ can be upper-bounded by

$$\tilde{T}_k(d') = 2 \cdot 9^{k} {\log d' + k \choose \log d'} - 1,$$
where $d'$ is the nearest power of two larger than $d$. We have $T_0(d) = 1$, so the base case holds. Let now $k > 0$. Let $H_1, H_2, \dots, H_j$ be the heavy paths traced by the \PrefixSearch for $P$, and $u_1, u_2, \dots, u_j$ be the positions where the search leaves the paths.

First, we estimate the number of \PrefixSearch operations in vertical $(k-1)$-errata trees that we run for the patterns of Type~\ref{tp:i}. Consider the weight-balanced tree for a heavy path $H_i$, and the path $\pi_i$ from the root of this tree to to $u_i$. The set of nodes covering the part of $H_i$ from its head to $u_i$ is the children of the nodes in $\pi_i$ hanging off to the left. Recall that each node of the weight-balanced tree has at most two children hanging off to the left. Since the weight of each $u_i$ is at least two times larger than the weight of $u_{i+1}$, and because of Fact~\ref{fact:wbt}, we obtain that the number of \PrefixSearch operations in vertical $(k-1)$-errata trees is at most $4 \left[T_{k-1}(d) + T_{k-1}(d/2) + \dots + 1 \right]$.

Analogously, we upper bound the number of \PrefixSearch operations in horizontal $(k-1)$-errata trees that we run for the patterns of Type~\ref{tp:ii} by $4 \left[ T_{k-1}(d) + T_{k-1}(d/2) + \dots + 1 \right]$. 

Finally, we upper bound the number of \PrefixSearch operations that we run in the same $k$-errata tree, i.e. from the nodes following $u_i$ in $H_i$ (these are the operations we run for the patterns of Types~\ref{tp:ii} and~\ref{tp:iii}) by $T_{k-1}(d) + T_{k-1}(d/2) + \dots + 1$. Summarizing, we obtain

$$T_k(d) \le 9 \left[T_{k-1}(d') + T_{k-1}(d'/2) + \dots + 1\right] +1 \le 9 \left[\tilde{T}_{k-1}(d') + \tilde{T}_{k-1}(d'/2) + \dots + 1\right] +1 .$$

Analogously to the space bound, we can show that $T_k(d) \le \tilde{T}_k(d') = \Oh(\log^k d)$. We also spend $\Oh(1)$ extra time for each \PrefixSearch operation to find the starting node of \PrefixSearch, which is $\Oh(\log^k d)$ time in total. Finally, we spend $\Oh(occ)$ time to output the desired dictionary strings.
\end{proof}

\rootedLCP*
\begin{proof}
At the preprocessing step, we traverse $T(\Suff)$ and remember the leftmost and the rightmost leaves in each of its subtrees. We also remember the neighbours of each leaf in the left-to-right order, and finally we preprocess the trie for the lowest common ancestor queries. We then create a predecessor data structure for each trie $\tau$ of the $k$-errata tree. By definition, $\tau$ is built on a subset of $\Suff$, and therefore, the rank of each string in $\tau$ relative to $\Suff$ is well-defined. The predecessor data structure is built on the set of the ranks of the strings in $\tau$. We also preprocess each trie $\tau$ for lowest common ancestor and weighted level ancestor queries. 

Suppose we would like to answer a rooted \PrefixSearch query for a string $Q$ and a trie $\tau$ of the $k$-errata tree. As we store the leftmost and the rightmost leaves in each subtree of $T(\Suff)$, we can find the predecessor and therefore the successor of $Q$ in $\Suff$ in $\Oh(1)$ time. Using the predecessor data structure, we can further find the predecessor $\pred$ and the successor $\successor$ of $Q$ in $\tau$ in $\Oh(\log d)$ time. We then compute the lowest common ancestor $u$ of $\pred$ and $\successor$ in $\Oh(1)$ time. We use two lowest common ancestor queries in $T(\Suff)$ to find the lengths $\ell_p, \ell_s$ of the longest common prefixes of $Q$ and $\pred$ and $Q$ and $\successor$. If $\ell_p = \ell_s$, we return $u$ as the answer. If $\ell_p > \ell_s$, then the answer is the ancestor of $\ell_p$ such that the length of its label is $\ell_p$. We can find it by one weighted level ancestor query. The case $\ell_s > \ell_p$ is analogous.
\end{proof}

\unrootedLCP*
\begin{proof}
The preprocessing step repeats that of Lemma~\ref{lm:rootedLCP}. We also assume to store the first letter on each edge of each trie. The search path for $Q$ traverses a number of heavy paths of~$\tau$. Consider the heavy path of $\tau$ containing $u$, and let $S$ be the label of the part of this heavy path starting from $u$ and up to the last node. If we know the end of the \PrefixSearch query for $Q$ in $T(\Suff)$, we can compute the length $\ell$ of the longest common prefix of $Q$ and $S$ using one lowest common ancestor query on $T(\Suff)$. We can then find the node in the path corresponding to the end of this longest common prefix using one weighted level ancestor query in $\Oh(\log d)$ time. From there, we can find the starting node of the second heavy path traversed by $Q$ in $\Oh(1)$ time. It remains to answer a rooted \PrefixSearch query for $Q' = Q[\ell+1,m]$ on the subtree rooted at this node, which belongs to a $(k-1)$-errata tree by construction. Knowing the answer for this subtree, we can go back to $\tau$ in $\Oh(\log d)$ time via one weighted level ancestor query. 
\end{proof}
\end{document}